\documentclass[conference]{IEEEtran}
\IEEEoverridecommandlockouts
\ifCLASSINFOpdf
\else
\fi

\usepackage{tabularx}
\usepackage{threeparttable}
\usepackage{mathrsfs}
\usepackage{color}
\usepackage{amsfonts,amssymb}
\usepackage{amsfonts}
\usepackage{subfigure}
\usepackage{booktabs}

\usepackage{amsmath}
\usepackage{enumerate}
\usepackage{algorithm}
\usepackage{algorithmic}
\usepackage{bm}
\usepackage{makecell}
\usepackage{multirow}
\usepackage{mathtools}
\usepackage{bbm}
\usepackage{dsfont}
\usepackage{pifont}
\usepackage{amsthm}
\usepackage{cite}
\usepackage{balance}

\usepackage{siunitx}

\newtheorem{remark}{Remark}
\newtheorem{theorem}{Theorem}
\newtheorem{lemma}{Lemma}
\newtheorem{corollary}{Corollary}
\newtheorem{assumption}{Assumption}

\newtheorem{definition}{Definition}

\begin{document}

\title{\vspace{15pt} 
Resistant Topology Inference in Consensus Networks: A Feedback-Based Design}
\author{Yushan Li, Jiabao He, and Dimos V. Dimarogonas
\thanks{This work was supported by the Knut and Alice Wallenberg (KAW) Foundation, and the Swedish Research
Council (VR). 
The work of Yushan Li was also supported by the Outstanding Ph.D. Graduate Development Scholarship from Shanghai Jiao Tong University.}
\thanks{The authors are with the Division of Decision and Control Systems, KTH Royal Institute of Technology, Stockholm, Sweden. Email: \{yushanl,jiabaoh,dimos\}@kth.se. }
}

\maketitle

\begin{abstract}
Consensus networks are widely deployed in numerous civil and industrial applications. 
However, the process of reaching a common consensus among nodes can unintentionally reveal the network’s topology to external observers by appropriate inference techniques. 
This paper investigates a feedback-based resistant inference design to prevent the topology from being inferred using data, while preserving the original consensus convergence. 
First, we characterize the conditions to preserve the original consensus, and introduce the \textit{“accurate inference”} notion, which accounts for both the uniqueness of the solution to topology inference (solvability) and the deviation from the original topology (accuracy). 
Then, we employ invariant subspace analysis to characterize the solvability. 
Even when unique inference remains possible, we provide necessary and sufficient conditions
for the feedback design to induce inaccurate inference, and give a Laplacian structure based distributed design. 
Simulations validate the effectiveness of the method. 
\end{abstract}

\IEEEpeerreviewmaketitle

\section{Introduction}

In the last decades, consensus networks are increasingly used to coordinate distributed systems such as power grids, social networks, and sensor networks, etc \cite{olfati2007consensus,10552333,wang2024distributed,10918814}. 
In these applications, the topology structure specifies the exchange of information among nodes to reach a common agreement, which is critical for system stability and convergence. 
However, the process of reaching consensus also faces a great risk of exposing the topology, which can be inferred by an external observer using collected data \cite{brugere2018network}. 
For instance, one can leverage the published data in social networks to infer a user's friendship network \cite{lu2020privacy}, and use the node trajectories to learn the topology of a multi-robot network in flocking tasks \cite{sebastian2023learning}. 
Such inference could lead to privacy breaches or even system vulnerabilities if the topology is exploited. 
Hence, developing resistant methods that can prevent the topology from being inferred is of great importance.

In the literature, topology inference (also called topology/network identification or reconstruction) has been a wildly investigated topic. 
The related works can be roughly cast into two categories according to the network system (NS) type: stochastic or deterministic. 
Typical methods for the former include graph signal processing \cite{zhu2020network}, power spectral analysis \cite{shahrampour2014topology}, and vector autoregressive analysis \cite{zaman2021online}, to name a few. 
These works mainly focus on the asymptotic inference performance with sufficient measurements on the NS. 
For deterministic networks, the efforts are mostly devoted to investigating the topology identifiability \cite{hayden2016network,shi2023single,sun2024identifiability} (i.e., whether the topology can be inferred given some data). 
For example, \cite{nabi-abdolyousefi2012networka} proposes a node knock-out method to reconstruct an undirected network. 
The work \cite{10337619} uses an edge-agreement-based framework to identify the topology of a class of nonlinear NSs. 
The identifiability conditions for general linear heterogeneous networks are analyzed in \cite{vanwaarde2021topology}. 

Compared with the topology inference research, how to construct its resistant design in a reverse way is relatively less studied. 
A first way to meet this goal is to introduce random disturbances in the iteration process. 
For example, the works \cite{katewa2020differential,Hawkins2024node} use differential privacy to guide the noise design when nodal data is published. 
However, their focus lies in the privacy of topology spectra, not the topology matrix itself. 
The work \cite{li2025Preserving} further develops a noise design for preserving the topology matrix, while without harming the converging state in the mean square sense. 
A common feature of these works is to introduce randomness into the original deterministic systems, thus obfuscating the topology for external observers. 
The tradeoff between topology preservation and control performance still remains an open issue. 
In fact, the topology identifiability analysis for deterministic networks can provide another perspective for the resistant design by enforcing the topology unidentifiability. 
To the best of our knowledge, this direction is far less discussed in the literature. 
A recent work \cite{mao2025unidentifiability} gives a novel unidentifiability notion based on the Fisher information matrix and presents a low-rank controller design to meet the unidentifiability for liner systems. 
However, this design also comes with a sacrifice of nominal control performance, e.g., state convergence.

In this paper, we focus on the resistant topology inference by feedback control design in deterministic consensus networks. 
This problem is challenging because i) the introduced control should well tackle the tradeoff between the resistant inference goal and the consensus preservation, 
and ii) the feedback needs to align with the distributed nature (or sparsity pattern) of the network topology. 

The contributions of this work are threefold. 
First, we obtain the conditions to preserve the consensus state when the feedback input is involved. 
Accordingly, we introduce \textit{``accurate inference''}, which takes the uniqueness of the inferred solution (solvability) and its deviation from the underlying topology (accuracy) into account, to evaluate a topology estimator.  
Second, we provide an invariant subspace based perspective to analyze the solvability of a nominal topology inference problem. 
Third, even if the solvability of a topology estimator is still preserved, we establish necessary and sufficient conditions for the feedback design to incur inaccurate inference. 
We further resemble the Laplacian structure to propose a fully distributed feedback design, while without disturbing the consensus convergence. 
Numerical experiments are conducted to corroborate the proposed method. 

The remainder of this paper is organized as follows. 
In Section \ref{sec:Preliminaries}, some preliminaries and the problem of interest are presented. 
The conditions and design for resistant topology inference are given in Section \ref{sec:results}. 
Numerical results are shown in Section \ref{sec:simulation}.
Finally, Section \ref{sec:conclusion} concludes the paper.

\section{Preliminaries and Problem of Formulation}\label{sec:Preliminaries}

\subsection{Graph Basics and Notations}
Consider a NS described by a digraph $\mathcal{G}=(\mathcal{V},\mathcal{E})$, where the vertex set $\mathcal{V}=\{1, \cdots, n\}$ enumerates $n$ interacting nodes, and the edge set $\mathcal{E} \subset \mathcal{V} \times \mathcal{V}$ characterizes directional edges. 
The presence of an edge $(i,j) \in \mathcal{E}$ signifies that node $i$ receives data from node $j$. 
The topology matrix of $\mathcal{G}$ is characterized by $W = [w_{ij}] \in \mathbb{R}^{n \times n}$, where $w_{ij} > 0$ denotes $(i,j) \in \mathcal{E}$, and $w_{ij} = 0$ otherwise. 
For node $i$, its incoming neighbors are defined as $\mathcal{N}_i = \{ j \in \mathcal{V} : (i,j) \in \mathcal{E} \}$. 

In this paper, we use $\otimes$ to denote the Kronecker product operator. 
Let $\bm{1}$ and $I$ be the all-one vector and identity matrix in compatible dimensions, respectively. 
The notation $\operatorname{rank}(\cdot)$ and $(\cdot)^\intercal$ denote the rank and transpose of a matrix, respectively. 
Given a series of vectors $v_1,\cdots,v_n$, denote $\operatorname{span}\{v_1,\cdots,v_n\}$ as the subspace spanned by these vectors, 
and let $\operatorname{dim}(\operatorname{span}\{v_1,\cdots,v_n\})$ be the number of independent vectors that span the subspace. 

\subsection{Network System and Topology Inference Model}
For distinction, we use the superscript $\cdot^*$ to denote the nominal NS without extra feedback, which is modeled as 
\begin{equation}\label{eq:local_model}
\begin{aligned}
x_{t+1}^{*,i} =  W_{ii} x_t^{*,i} + \sum\nolimits_{j\in \mathcal{N}_{i}} W_{ij} x_t^{*,j},
\end{aligned}
\end{equation}
where $x_{t}^{*,i}$ is the $i$-th node's state at time $t\in\{0,1,2,\cdots\}$. 
For the topology $W$, the following assumption is made. 
\begin{assumption}\label{assu:topo}
The topology matrix $W$ is row-stochastic, and $1$ is a simple eigenvalue while all other eigenvalues have modulus less than $1$.
\end{assumption}

By Assumption \ref{assu:topo}, the states of all nodes of \eqref{eq:local_model} will reach the following consensus point as $t\to\infty$ \cite[Theorem 5.1]{FB-LNS}
\begin{equation}
x^*_{\infty}=\lim_{t\to\infty} x_{t}^* = \lim_{t\to\infty} W^t x_0^* = (w^\intercal x_0^*) \bm{1},
\end{equation}
where $x_{t}^*=[x_{t}^{*,1},\cdots,x_{t}^{*,n}]^\intercal \in\mathbb{R}^n$ is the global system state, and $w\in\mathbb{R}^{n}$ is the normalized left eigenvector of the eigenvalue $1$ ($w^\intercal \bm{1}=1$). 

Suppose that an external observer can access consecutive $T\ge n$ steps of the NS'states, and the states satisfy 
\begin{equation}
    X_b^*=W X_a^*,
\end{equation}
where $X_a^*=[x_{0}^*,x_{1}^*,\cdots,x_{T-1}^*]$ and $X_b^*=[x_{1}^*,x_{2}^*,\cdots,x_{T}^*]$
Then, the topology $W$ can be directly inferred by the following ordinary least-squares (OLS) estimator
\begin{equation}\label{eq:old_estimator}
\hat{W}^*= X_b^* (X_a^*)^{\dag},
\end{equation} 
where $(\cdot)^\dag$ represents the pseudo-inverse of a matrix.

\subsection{Problem of Interest}
Considering that the topology inference of an external observer could lead to privacy breaches, this work aims to provide a feedback design to prevent the topology from being inferred. 
To address this problem, we use a state feedback control design for the state update process, given by
\begin{equation}\label{eq:local_feedback}
\left\{\begin{aligned}
u^{i}_t &=\sum\nolimits_{j\in\mathcal{N}_i} K_{ij} x^{j}_t, \\
x_{t+1}^{i} &=  W_{ii} x_t^{i} + \sum\nolimits_{j\in \mathcal{N}_{i}} W_{ij} x_t^{j} + u^{i}_t,
\end{aligned}\right.
\end{equation}
where $x_0=x_0^*$ and $K=[K_{ij}]_{i,j=1}^{n}\in\mathbb{R}^{n\times n}$ is the feedback matrix. 
The above design can be regarded as adding an extra input to the nominal model \eqref{eq:local_model}. 
Notice that the input $u^{i}_t$ needs to be well designed to preserve the consensus point of the original NS \eqref{eq:local_model}, which is crucial for the reliability of consensus algorithms in NSs. 
Then, the global form of \eqref{eq:local_feedback} is written as
\begin{equation}\label{eq:global_feedback}
x_{t+1}= (W+K) x_t \triangleq \tilde{W} x_t. 
\end{equation}
Then, the external observer will have the following OLS estimator based on the collected $T$ steps of states
\begin{equation}\label{eq:W_estimator}
\hat{W}= X_b (X_a)^{\dag},
\end{equation} 
where $X_a=[x_{0},x_{1},\cdots,x_{T-1}]$ and $X_b=[x_{1},x_{2},\cdots,x_{T}]$. 
To formally describe the estimated topology, we introduce the following definition. 
\begin{definition}
The topology estimator $\hat{W}$ is accurate if i) \eqref{eq:W_estimator} has a unique solution, i.e., $\operatorname{rank}(X_a)=n $, and ii) the estimator has no deviation with $W$, i.e., $\|\hat{W}-W\|=0$. 
\end{definition}

Based on the above formulation, our problem of interest can be cast as below
\begin{subequations}\label{eq:prob}
\begin{align}
\text{Find} &~~K\\
\text{s.t.} & ~~\hat{W}\neq W,\label{eq:requirement1} \\
&~~\lim_{t\to\infty} x_{t}= x^*_{\infty} = (w^\intercal x_0^*) \bm{1}, \label{eq:converge_state} \\
&~~K_{ij}=0~\text{if}~W_{ij}=0~(i\neq j). \label{eq:K_sparse}
\end{align}
\end{subequations}
The necessities of the above constraints are as follows. 
\begin{itemize}
    \item First, \eqref{eq:requirement1} describes the core goal to prevent the underlying topology from being inferred. 
    \item Second, \eqref{eq:converge_state} requires that the consensus point of the original NS \eqref{eq:local_model} needs to be preserved. 
    \item Third, the sparsity constraint \eqref{eq:K_sparse} is enforced due to the distributed nature of a consensus network, where each node can only access information from its neighbors. 
\end{itemize}

To tackle problem \eqref{eq:prob}, our main idea is to exploit conditions under which an accurate inference will be impossible. The proposed methods fall into two categories: i) let \eqref{eq:W_estimator} have no unique solution (degenerate the solvability), or ii) make the solution deviated with $W$ (inaccurate inference) even if a unique solution for \eqref{eq:W_estimator} is available. 

\section{Main Results}\label{sec:results}


\subsection{Conditions for Consensus Preservation}

First, we show how the consensus preservation condition \eqref{eq:converge_state} can be characterized in terms of $K$. 
Hereafter, we denote $\lambda_i$ and $\mu_i$ as $i$-th eigenvalue of $W$ and $\tilde{W}$, respectively. 
\begin{lemma}\label{le:K_conditions}
Under Assumption \ref{assu:topo}, the consensus preservation condition \eqref{eq:converge_state} can be achieved if and only if $K$ satisfies
\begin{align}
    &K \bm{1}=0, \label{eq:K_c1}\\
    &w^\intercal K=0,\label{eq:K_c2}\\
    &|\mu_i|<1,~i=2,\cdots,n.  \label{eq:K_c3}
\end{align}
\end{lemma}
\begin{proof}
\textit{Necessity}: 
First, if \eqref{eq:converge_state} holds, the system state will no longer change when it reaches $x^*_{\infty}$. 
Hence, we have
\begin{equation}
   \tilde{W}x^*_{\infty}= (W+K)x^*_{\infty}=x^*_{\infty} \Rightarrow K \bm{1}=0, 
\end{equation}
which indicates that $\tilde{W}$ has an eigenvalue $1$ (i.e., $\mu_1=1$). 
Meanwhile, recall that $x^*_{\infty}$ is determined by the left vector $w$ corresponding to $\lambda_1=1$ in $W$. 
To ensure that $w$ remains a left eigenvector of $\tilde{W}$ with eigenvalue $\mu_1=1$, we need 
\begin{equation}
w^\intercal \tilde{W}=w^\intercal \Rightarrow w^\intercal K=0. 
\end{equation}
Finally, since $\tilde{W}^t$ is convergent, it indicates that the eigenvalues of $\tilde{W}$ except 1 need to be within the unit circle, hence yielding \eqref{eq:K_c3}. 
The necessity is completed. 

\textit{Sufficiency}: If the conditions \eqref{eq:K_c1}, \eqref{eq:K_c2}, and \eqref{eq:K_c3} hold for $\tilde{W}$, one can directly resort to the Jordan decomposition of $\tilde{W}$ to obtain the limit $\lim_{t\to\infty} x_t= \lim_{t\to\infty}\tilde{W}^t =x^*_{\infty}$ (see \cite[Ch 2.1]{FB-LNS} for more details). 
This completes the proof. 
\end{proof}

The conditions in Lemma \ref{le:K_conditions} guarantee that $x_{t}$ can exactly converge to the original consensus point $x_{\infty}^*$. 
Next, we will show that satisfaction of equality constraints \eqref{eq:K_c1} and \eqref{eq:K_c2} can guarantee a feasible $K$ that also satisfies \eqref{eq:K_c3}. 

\begin{theorem}\label{th:eigenvalue}
Suppose there exists a $K_0\in\mathbb{R}^{n\times n}$ such that $K_0 \bm{1}=0$ and $w^\intercal K_0=0$. 
Then, a feasible $K$ can always be found such that $|\mu_i|<1, i=2,\cdots,n,$ holds. 
\end{theorem}

\begin{proof}
To ease analysis, we assume that $W$ is diagonalizable and prove the non-diagonalizable case in the end. 

First, the eigenvalue decomposition for $W$ is given by 
\begin{equation}
W=V\Lambda V^{-1},
\end{equation}
where $\Lambda=\operatorname{diag}(\lambda_1,\cdots,\lambda_n)$ and the $i$-th column of $V$ is the right eigenvector corresponding to $\lambda_i$. 
Then, we have that $\lambda_1=1$ and the other eigenvalues $\lambda_i,i=2,\dots,n$ satisfy
\begin{equation}\label{eq:gap}
|\lambda_i|\le 1-\delta \Rightarrow \delta \le 1-|\lambda_i| \le |1 - \lambda_i|
\end{equation}
for some $\delta>0$. 
This means there is at least a gap of $\delta$ between 1 and all other eigenvalues for $W$. 
Then, applying the Bauer-Fike theorem \cite{bauer1960norms}, we obtain that for any eigenvalue of $\tilde{W}$, there exists an eigenvalue $\lambda_s$ of $W$ such that 
\begin{equation}\label{eq:mu_bound}
|\mu_i-\lambda_s|\le \kappa(V)\|K\|,~i=1,\cdots,n,
\end{equation}
where $\kappa(V)=\|V\|\|V^{-1}\|$ is the condition number of $V$. 
Since $K_0 \bm{1}=0$ and $w^\intercal K_0=0$, 
we can always find a small $\epsilon>0$ to construct a new feedback matrix $K=\epsilon K_0$, such that 
\begin{equation}\label{eq:K_bound}
\|K\|=\epsilon\|K_0\|<\frac{\delta}{2\kappa(V)}. 
\end{equation}
Then, we turn to prove that the eigenvalue $\lambda_s\neq 1$ in \eqref{eq:mu_bound}. 

When \eqref{eq:K_bound} holds, assume for contradiction that there exists a non-$1$ eigenvalue $\mu_{i_0}$ of $\tilde{W}$ such that $|\mu_{i_0}-1|<\frac{\delta}{2}$. 
Substituting \eqref{eq:K_bound} into \eqref{eq:mu_bound}, $\mu_{i_0}$ should also satisfy 
\begin{equation}\label{eq:mu_22}
|\mu_{i_0}-\lambda_s|<\frac{\delta}{2}. 
\end{equation}
There are two possibilities for $\lambda_s$ in \eqref{eq:mu_22}. 
\begin{itemize}
\item If $\lambda_s=1$, we have $|\mu_{i_0}-1|<\frac{\delta}{2}$. 
However, recall that $1$ is a simple eigenvalue of $W$ and it is also preserved in $\tilde{W}$. 
By classical perturbation theory, a simple eigenvalue persists uniquely under sufficiently small perturbations (see \cite[Ch 2.1.8]{kato2013perturbation}). 
In other words, if $\mu_{i_0}$ is distinct from $1$ yet satisfying $|\mu_{i_0}-1|<\frac{\delta}{2}$, it will violate the uniqueness of the eigenvalue 1 in $\tilde{W}$. 
\item If $\lambda_s\neq 1$, it should satisfy the gap \eqref{eq:gap}, i.e., $|1 - \lambda_s| \ge \delta $. 
Applying the triangle inequality, we obtain
\begin{align}
|\mu_{i_0}-1|&=|\mu_{i_0}-\lambda_s + \lambda_s-1| \nonumber \\
&\ge |\lambda_s-1| - |\mu_{i_0}-\lambda_s|  > \delta - \frac{\delta}{2}=\frac{\delta}{2},
\end{align}
which contradicts with the assumption $|\mu_{i_0}-1|<\frac{\delta}{2}$. 
\end{itemize}
In either case, a contradiction arises. 
Hence, we have that for arbitrary eigenvalue $ \mu_{i_0} \neq 1$ in $\tilde{W}$, the corresponding $\lambda_s \neq 1 $. 
Following this conclusion, $\mu_{i_0}$ is upper bounded as 
\begin{align}
|\mu_{i_0}| &= |\mu_{i_0} - \lambda_s + \lambda_s| \le |\mu_{i_0}-\lambda_s| + |\lambda_s| \nonumber \\
&\le \frac{\delta}{2} + (1-\delta)=1-\frac{\delta}{2}.
\end{align}
This shows that all eigenvalues of $\tilde{W}$ except $1$ remain strictly within the unit circle, which indicates that \eqref{eq:K_c3} can hold. 

As for the case where $W$ is non-diagonalizable but still satisfies Assumption \ref{assu:topo}, we can leverage generalized Bauer-Fike theorem \cite{bauer1960norms} to establish a more complex upper bound for $|\mu_i-\lambda_s|$. 
The following proof procedures resemble those after \eqref{eq:mu_bound} and are omitted here. 
\end{proof}

Theorem \ref{th:eigenvalue} demonstrates that as long as the equality constraints \eqref{eq:K_c1} and \eqref{eq:K_c2} have non-zero solutions, then a feasible $K$ that meets the eigenvalue modulus constraint \eqref{eq:K_c3} always exists. 
This result is guaranteed by the continuity of eigenvalues of a matrix under disturbances. 
With this in mind, we only need to focus on the equality constraints in the following analysis for finding appropriate feedback $K$. 



\subsection{Invariant Subspace based Solvability Analysis}


The part investigates the solvability of \eqref{eq:W_estimator}. 
The analysis involves the following definitions and properties. 

\begin{definition} {\cite[Ch. 1]{Gohberg2006invariant}}
    Let $A: \mathbb{R}^{n} \to \mathbb{R}^{n}$ be a linear transformation. A subspace $\mathcal{A} \subset \mathbb{R}^{n}$ is called $A$-invariant, if $Ax \in \mathcal{A}$ for every vector $x \in \mathcal{A}$.
\end{definition}

Trivial examples of invariant subspaces are $\left\{0\right\}$ and $\mathbb{R}^{n}$. 
The subspaces $\operatorname{Ker}(A^m) = \left\{x \in \mathbb{R}^{n} | A^mx=0\right\}$ and $\operatorname{Im}(A^m) = \left\{A^mx| x \in \mathbb{R}^{n} \right\}$, where $m=1,2,...$, are also $A$-invariant. 

\begin{lemma}{\cite[Ch. 1]{Gohberg2006invariant}} \label{le:condition}
Given a transformation $A$ and $x\in \mathbb{R}^{n}$, 
the subspace $\mathcal{A}_\infty = \operatorname{span}\left\{x,Ax,\cdots\right\}$ is equivalent to $\mathcal{A}_n = \operatorname{span}\left\{x,Ax,\cdots,A^{n-1}x\right\}$, and $\mathcal{A}_n$ is $A$-invariant.
\end{lemma}

The above lemma suggests that to study the rank of $X_a = [x_0,\tilde{W}x_0,\tilde{W}^2x_2,\cdots,\tilde{W}^{T-1}x_0]$ ($T\ge n$), it is equivalent to investigate the dimension of the following subspace
\begin{equation} \label{eq:Inv_W}
\tilde{\mathcal{W}}_n \triangleq  \operatorname{span}\left\{x_0,\tilde{W}x_0,\tilde{W}^2x_0,\cdots,\tilde{W}^{n-1}x_0\right\}.
\end{equation}
Then, it holds that $\operatorname{dim}(\tilde{\mathcal{W}}_n)=\operatorname{dim}(\tilde{\mathcal{W}}_{\infty})=\operatorname{rank}(X_a)$. 

\begin{lemma} \label{lemma:solvability}
    Let $\mathcal{S}(\tilde{W})$ be the set of all proper invariant subspaces in $\tilde{\mathcal{W}}_n$,  where ``proper'' excludes the trivial subspaces $\left\{0\right\}$ and $\mathbb{R}^{n}$. Then, we have 
        $$\operatorname{rank}(X_a) = n~\text{if and only if}~x_{0}\notin\mathcal{S}(\tilde{W}).$$
\end{lemma}
\begin{proof}
\textit{Necessity}: We prove the necessity by contraposition: 
if $x_{0} \in \mathcal{S}(\tilde{W})$, then $\operatorname{dim}(\tilde{\mathcal{W}}_n) < n$. 
Denote by $\mathcal{S}_{0} \subseteq \mathcal{S}(\tilde{W})$ a subset where $x_{0}$ lies in. 
Since $\mathcal{S}_{0}$ is a proper $\tilde{{W}}$-invariant subspace that excludes $\left\{0\right\}$ and $\mathbb{R}^{n}$ by definition, we have $\operatorname{dim}(\mathcal{S}_{0}) < n$ and $\tilde{{W}}^{t} x_{0}\in\mathcal{S}_{0}$, $\forall t \ge 0$.  
Hence, each vector in $\{x_{0}, \tilde{{W}} x_{0}, \tilde{{W}}^{2}x_{0}, \dots, \tilde{{W}}^{n-1} x_{0}\}$ 
lies in $\mathcal{S}_{0}$, implying $\tilde{\mathcal{W}}_n $ defined in \eqref{eq:Inv_W} is a subset of $\mathcal{S}_{0}$. 
It then follows that $\operatorname{dim}(\tilde{\mathcal{W}}_n) \le \operatorname{dim}(\mathcal{S}_{0}) <n$.

\textit{Sufficiency}: Consider $x_{0}\notin \mathcal{S}_{0}$ for every $\mathcal{S}_{0}\subseteq\mathcal{S}(\tilde{W})$. 
This means no proper invariant subspace contains $x_{0}$. 
If $\tilde{\mathcal{W}}_n$ itself were contained in a proper invariant subspace, then $x_{0}$ would also lie there, contradicting with $x_{0}\notin \mathcal{S}_{0}$. 
Hence, $\tilde{\mathcal{W}}_n$ cannot be contained in any proper invariant subspace, indicating $\tilde{\mathcal{W}}_n =\mathbb{R}^n$ and $\operatorname{dim}(\tilde{\mathcal{W}}_n) = n$. 
\end{proof}

Similar to Lemma \ref{lemma:solvability}, for the nominal NS \eqref{eq:local_model}, if $x_0 \in \mathcal{S}(W)$, where $\mathcal{S}(W)$ is the set of all proper invariant subspaces of $W$, then it indicates that $\operatorname{rank}(X_a^*) < n$ and the topology estimator \eqref{eq:old_estimator} is already unsolvable. 
Therefore, our primary focus lies in investigating that when $x_0 \notin \mathcal{S}(W)$, 
whether it is possible to design a feedback matrix $K$ such that $\operatorname{rank}(X_a) < n$, rendering the topology estimator \eqref{eq:W_estimator} unsolvable. 
The following theorem demonstrates that this is indeed attainable.

\begin{theorem} \label{th:solvability}
    For any non-zero $x_0 \notin \mathcal{S}(W)$, there exists a feedback matrix $K$ such that $\operatorname{rank}(X_a) < n$.
\end{theorem}
\begin{proof}
We prove this existence statement by construction. 
First, recall that $\bm{1}$ and $w$ are the right and left eigenvectors corresponding to $\lambda_1=1$ for $W$. 
Let $K = pq^\intercal$ for two vectors $p,q \in \mathbb{R}^{n}$, where the vectors satisfy 
\begin{equation}\label{eq:pq}
q^\intercal\bm{1}=0,~q^\intercal x_0 \neq 0,~\text{and}~w^\intercal p=0.
\end{equation}
By construction \eqref{eq:pq}, the conditions \eqref{eq:K_c1} and \eqref{eq:K_c2} are ensured. 

Next, we introduce a scalar $\beta$ to enable $\tilde{W}x_0 = \beta x_0$, such that $\operatorname{rank}(X_a) < n$. 
Then, it follows from $Wx_0 + (pq^\intercal) x_0 =  \beta x_0$ that $p$ can be solved as
\begin{equation}\label{eq:p_solu}
p = \frac{\beta x_0-Wx_0}{q^\intercal x_0}. 
\end{equation}
Meanwhile, to ensure $w^\intercal p=0$, we have
\begin{equation}\label{eq:beta}
    w^\intercal \frac{\beta x_0-Wx_0}{q^\intercal x_0}=0 \Rightarrow \beta = \frac{w^\intercal Wx_0}{w^\intercal x_0},
\end{equation}
where $w^\intercal x_0$ is exactly the consensus value. 
Hence, constructing $K$ by \eqref{eq:beta}, \eqref{eq:pq} and \eqref{eq:p_solu} sequentially can make $\operatorname{rank}(X_a) < n$. 
The proof is completed. 
\end{proof}

Theorem \ref{th:solvability} indicates that one can always construct a feedback matrix $K$ that renders $\hat{W}$ unsolvable. 
Specifically, the proof process provides a rank‐$1$ construction approach to achieve so, whose procedure is summarized below. 
\begin{itemize}
    \item Step 1: pick any $q$ such that $q^\intercal\bm{1}=0$ and $q^\intercal x_0 \neq 0$. 
    \item Step 2: take $p = \frac{\beta x_0-Wx_0}{q^\intercal x_0}$, where $\beta$ is given in \eqref{eq:beta}.
    \item Step 3: construct $K = pq^\intercal$. 
\end{itemize}
However, this approach has not yet explicitly accounted for the constraints \eqref{eq:converge_state} and \eqref{eq:K_sparse}. 
Whether a nontrivial solution exists under all the constraints is more involved than the setting of Theorem \ref{th:solvability}, as a nonlinear factorization type constraint is incurred by the rank‐$1$ construction. 
We leave related analysis and design in this direction for future research.

\subsection{Conditions of Achieving Inaccurate Inference}
Despite the difficulty of degenerating the solvability for $\hat{W}$, in this part, we show how to invalidate the accurate inference by allowing a unique solution for $\hat{W}$ but making it inaccurate. 
To begin with, we show the existence condition of a $K$ satisfying the sparsity constraint \eqref{eq:K_sparse}.

Treating \eqref{eq:K_c1} and \eqref{eq:K_c2} as two matrix equations on $K$, they can be vectorized as 
\begin{equation}\label{eq:M_constraint}
\left\{\begin{aligned}
(\bm{1}_n^\intercal \otimes I_n )\operatorname{vec}(K)=0\\
(I_n \otimes w^\intercal)  \operatorname{vec}(K)=0
\end{aligned}
\right.
\Rightarrow  M \theta_K=0,
\end{equation}
where $M=[(\bm{1}_n^\intercal \otimes I_n )^\intercal, (I_n \otimes w^\intercal)^\intercal]^\intercal \in \mathbb{R}^{2n \times n^2}$ and $\theta_K = \operatorname{vec}(K)  \in \mathbb{R}^{n^2}$. 
Considering the sparsity constraint \eqref{eq:K_sparse}, we define the index set for non-zero elements in $W$ as 
\begin{equation}\label{eq:def_z}
\mathcal{Z}=\{W_{ij}: W_{ij} \neq 0,i,j\in\mathcal{V}\},
\end{equation}
with $z=|\mathcal{Z}|$ being the cardinality number of $\mathcal{Z}$. 
Then, denote $\tilde{\theta}_K \in\mathbb{R}^{z}$ as the reduced parameter vector of $\theta_K$ after deleting zero elements in $\theta_K$, and \eqref{eq:M_constraint} can be reformulated as 
\begin{equation}\label{eq:new}
\tilde{M}  \tilde{\theta}_K=0,
\end{equation}
where $\tilde{M}\in\mathbb{R}^{2n \times z}$ is a modified version of $M$ by deleting the columns of $M$ corresponding to zeros elements in $\theta_K$. 
Next, we denote by $\mathcal{K}$ the set of all feasible $K$: 
\begin{equation}
\mathcal{K}\!=\!\{K\!\in\!\mathbb{R}^{n\times n}: K~\text{satisfies}~\eqref{eq:K_c1},~\eqref{eq:K_c2},~\eqref{eq:K_c3}~\text{and}~\eqref{eq:K_sparse}\},
\end{equation}
and present existence conditions for $\mathcal{K}$ as below. 
\begin{theorem}\label{th:exist_K}
The set $\mathcal{K} \neq \emptyset$ if the following condition holds
\begin{equation}\label{eq:condition}
\operatorname{rank}(\tilde{M})< z.  
\end{equation}
\end{theorem}
   
\begin{proof}
For the homogenous linear equations \eqref{eq:new}, we have
\begin{itemize}
\item  if $\operatorname{rank}(\tilde{M})>z$, \eqref{eq:new} has no solution;
\item if $\operatorname{rank}(\tilde{M})=z$, \eqref{eq:new} has only trivial solution $\tilde{\theta}_K=0$. 
\end{itemize}
Hence, \eqref{eq:new} has nontrivial solutions only when $\operatorname{rank}(\tilde{M})<z$. 
Let $\tilde{\theta}_{K^*}$ be a nontrivial solution to equation \eqref{eq:new} when $\operatorname{rank}(\tilde{M})< z$ holds. 
The augmented version of $\tilde{\theta}_{K^*}$, $\theta_{K^*}$, is obtained by supplementing the zeros elements required by the sparsity constraint \eqref{eq:K_sparse}. 
Then, we restore the solution matrix $K^*$ by reversing the vectorization \eqref{eq:M_constraint}. 
It is easy to verify that $K^*$ satisfies the equality conditions \eqref{eq:K_c1},~\eqref{eq:K_c2}~and~\eqref{eq:K_sparse}.

Finally, based on Theorem \ref{th:eigenvalue}, we can always find a small $\epsilon>0$ to construct $K=\epsilon K^*$, such that $|\mu_i|<1,~i=2,\cdots,n$ holds, 
which verifies $\mathcal{K}\neq \emptyset$ and completes the proof. 
\end{proof}

Theorem \ref{th:exist_K} demonstrates the existence of a feasible feedback matrix to make the topology estimator inaccurate. 
This feasibility comes from two aspects: 
i) the sparsity structure of $W$ can allow more free parameters than the number of independent constraints; 
ii) as long as \eqref{eq:condition} holds, one can always find a small enough $K$ to satisfy the eigenvalue inequality \eqref{eq:K_c3}, which is guaranteed by the continuity of eigenvalues under small disturbances.

\begin{remark}\label{rema:solution}
We observe that Theorem \ref{th:exist_K} also provides the idea of obtaining an appropriate $K$. 
A straightforward way is to 
\begin{itemize}
\item i) obtain the basis of the kernel space of $\tilde{M}$;
\item ii) combine the basis vectors to form a candidate $K_0$ given desired criteria (e.g., making $\|K^*\|_0$ maximum); 
\item iii) scale $K\!=\!\epsilon K_0 (\epsilon>0)$ such that $|\mu_i|\!<\!1,i=2,\!\cdots\!,n$.
\end{itemize}
This solving manner is suitable for offline or centralized situations, where $K$ needs to be assigned for each node. 
\end{remark}


\subsection{Distributed Design for Inaccurate Inference}

Although finding an optimal feedback matrix $K$ without relying on global information is challenging in general, in this part we will show that it is feasible to construct a certain class of feedback $K$ without global knowledge.

First, note that the critical properties of $K$ lie in i) preserving the left and right eigenvectors of $W$, and ii) respecting the sparsity pattern of $W$ to enable a distributed implementation. 
Since the matrix $W$ itself exhibits the desired structure, we present a distributed design as
\begin{equation}\label{eq:local_K}
K_{ij}=\left\{\begin{aligned}
&\alpha W_{ij},&&\text{if}~i\neq j \\
&\alpha (W_{ii}-1),&&\text{if}~i=j
\end{aligned}\right.,
\end{equation}
where $\alpha > 0$ a positive scalar gain. 
It is clear that $K_{ij}$ in \eqref{eq:local_K} strictly subjects to the constraint \eqref{eq:K_sparse} and indicates that node $i$ can only use information from its neighbors and itself. 
In a global form, $K$ can be written as 
\begin{equation}\label{eq:K_design}
K = -\alpha (I-W),
\end{equation}
which can be regarded a scaled negative Laplacian of the graph $\mathcal{G}$. 
Next, we present the following result. 
\begin{corollary}\label{coro:double}
For the eigenvalues of $W$, denote $r_{\max}=\max\{|\lambda_i|,i=2,\cdots,n\}$. 
The feedback $K$ by \eqref{eq:K_design} can satisfy all constraints in problem \eqref{eq:prob}, if the gain $\alpha$ satisfies
\begin{equation}\label{eq:alpha_bound}
\alpha < \frac{1 - r_{\max}}{1 + r_{\max}}.
\end{equation}
\end{corollary}

\begin{proof}
First, when $W$ satisfies Assumption \ref{assu:topo} and $K$ is designed by \eqref{eq:K_design}, we have 
\begin{equation}
\left\{\begin{aligned}
&K \bm{1} = -\alpha (I-W)\bm{1} = \alpha(\bm{1} -\bm{1})=  0 \\
&w^\intercal K = -\alpha w^\intercal(I-W) = \alpha (w^\intercal-w^\intercal)= 0
\end{aligned}
\right.,
\end{equation}
which indicates that the equalities \eqref{eq:K_c1} and \eqref{eq:K_c2} are satisfied. 

Next, we focus on the eigenvalues of $\tilde{W}$. 
Since $\tilde{W}=(1+\alpha)W-\alpha I$, its eigenvalues can be explicitly written as 
\begin{equation}
\mu_i= (1+\alpha)\lambda_i-\alpha,~i=1,\cdots,n,    
\end{equation}
where $\mu_1=\lambda_1=1$ by construction. 
Denote the difference between modulus $|\mu_i|$ and 1 as $g(\lambda_i)=|\mu_i|-1$, and we have
\begin{align}\label{eq:upper}
g(\lambda_i) &=|(\lambda_i-1)\alpha + \lambda_i|-1 \le |\lambda_i-1|\alpha + |\lambda_i| -1.
\end{align}
If the upper bound in \eqref{eq:upper} is enforced to be smaller than zero, then $|\mu_i|<1,i=2,\cdots,n,$ always hold. 
Then, it follows that
\begin{equation}\label{eq:upper2}
|\lambda_i-1|\alpha + |\lambda_i| -1<0 \Leftrightarrow  \alpha<\frac{1-|\lambda_i|}{|\lambda_i-1|}. 
\end{equation}
Note that the right-hand-side term of the second inequality in \eqref{eq:upper2} can be lower bounded by
\begin{equation}
\frac{1-|\lambda_i|}{|\lambda_i-1|}\ge \frac{1-|\lambda_i|}{|\lambda_i|+1} \ge \frac{1-r_{\max}}{r_{\max}+1}. 
\end{equation}
If we substitute $\alpha<\frac{1-r_{\max}}{r_{\max}+1}$ into \eqref{eq:upper}, it yields that 
\begin{align}
g(\lambda_i) & <  |\lambda_i-1| \frac{1-r_{\max}}{r_{\max}+1} + |\lambda_i| -1  \nonumber \\
& \le 1-r_{\max}+|\lambda_i| -1 \le 0,~i=2,\cdots,n,
\end{align}
which ensures that the eigenvalue modulus condition is satisfied. 
The proof is completed.  
\end{proof}

Note that \eqref{eq:alpha_bound} is just a sufficient condition to meet the constraint \eqref{eq:K_sparse}. 
It indicates that as long as $\alpha$ is small enough, the feedback $K$ can always meet our requirement, which corresponds to the conclusion in Theorem \ref{th:exist_K}. 
More generally, as shown in the proof of Corollary \ref{coro:double}, any $\alpha$ that satisfies $|(\lambda_i-1)\alpha + \lambda_i|<1 (i=2,\cdots,n) $ will work for the design.

\section{Numerical Simulations}\label{sec:simulation}
 
In this part, we give numerical examples to verify the theoretical results. 
First, we randomly generate the following direct topology matrix for a network of $n=6$ nodes
\begin{small}
\begin{equation*}
W= 
\begin{bmatrix}
\num{0.2559} & \num{0.0648} & \num{0.2437} & \num{0.2308} & \num{0.0618} & \num{0.1430} \\
\num{0.2677} & \num{0.0203} & \num{0.0848} & \num{0.2107} & \num{0.2532} & \num{0.1633} \\
\num{0.0416} & \num{0.2094} & \num{0.1535} & \num{0.2865} & \num{0.2366} & \num{0.0724} \\
\num{0.2740} & \num{0.1388} & \num{0.3221} & \num{0.2221} & \num{0.0245} & \num{0.0185} \\
\num{0.0910} & \num{0.2475} & \num{0.2866} & \num{0.1932} & \num{0.1240} & \num{0.0577} \\
\num{0.2022} & \num{0.1368} & \num{0.0539} & \num{0.2817} & \num{0.2879} & \num{0.0375}
\end{bmatrix}.
\end{equation*}
\end{small}
\!\!It can be verified that $W$ satisfies Assumption \ref{assu:topo}, and the second largest eigenvalue modulus is given by $r_{\max}=0.147$. 
Then, for the homogeneous equation $\tilde{M}  \tilde{\theta}_K=0$ corresponding to this case, we have $\tilde{\theta}_K\in\mathbb{R}^{36}$ and $\operatorname{rank}(\tilde{M})=11$. 
Hence, there are infinite solutions for $K$ that can make $\hat{W}$ inaccurate, e.g., by the procedures in Remark \ref{rema:solution}. 
Next, we turn to examine the distributed design \eqref{eq:local_K} for $K$ to achieve inaccurate inference. 
Here, the critical value for $\alpha$ is given by $({1 - r_{\max}})/({1 + r_{\max}})=0.7437$ based on the condition \eqref{eq:alpha_bound}. 
We consider 5 groups of $\alpha$, $\{0.14,0.28,0.42,0.56,0.7\}$, to construct $K$. 
The following metrics are used to evaluate the state error and inference error under $\hat{W}$, respectively 
\begin{align}
E_1(t)=\|x_t-x^*\|_2, ~E_2(\hat{W})=\|\hat{W}-W\|_2. 
\end{align}
Meanwhile, we consider the popular \emph{Jaccard index} to quantify the change in the sparsity of $\tilde{W}$, 
where a higher Jaccard index (closer to 1) indicates that $\tilde{W}$ is more similar to $W$ in sparsity.  



\begin{figure}[t]
\centering
\includegraphics[width=0.33\textwidth]{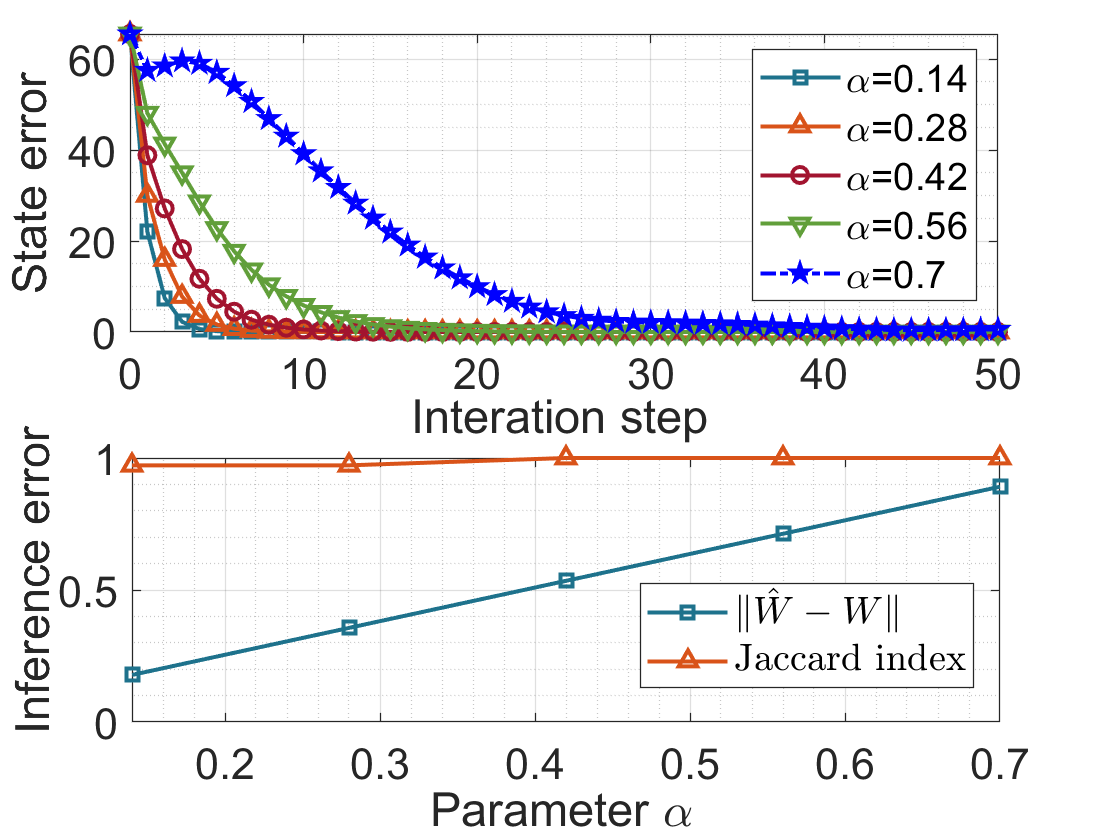}
\vspace{-8pt}
\caption{State convergence error and inference performance under $\tilde{W}$}
\label{fig:all_error}
\vspace{-10pt}
\end{figure}

The simulation results are plotted in Fig.~\ref{fig:all_error}. 
In the upper subfig, the state convergence error $E_1(t)$ in $T=50$ steps is presented. 
It is clear that in all settings of $\alpha$, the errors will approach zero asymptotically. 
As $\alpha$ increases, $\lambda_2(\tilde{W})$ becomes larger, thus incurring a slower convergence speed. 
The bottom subfigure depicts the inference performance of $\hat{W}$. 
It is evident that the inference error $E_2(\hat{W})$ will monotonically increase as $\alpha$ grows. 
This matches our intuition that $\alpha$ directly determines the inference bias. 
For the metric $\operatorname{Jaccard}(\mathcal{Z}_W, \mathcal{Z}_{\tilde{W}})$ where we set $\varepsilon_0=0.2$, the curve is very close or equals to 1. 
For the non-one points in this curve, they occur because some diagonal elements $\tilde{W}_{ii}=(1+\alpha)W_{ii}-\alpha$ are extremely small and not counted, while the sparsity of non-diagonal elements is still unchanged. 
Hence, the sparsity constraint \eqref{eq:K_sparse} is met.

\section{Conclusions}\label{sec:conclusion}

In this paper, we have investigated the problem of feedback control designs that resist topology inference in consensus networks without disrupting consensus convergence. 
First, we introduced the notion of ``accurate inference'', which integrates both the uniqueness (solvability) and the deviation (accuracy) of topology inference. 
Then, through invariant subspace analysis, we established the conditions under which the feedback designs can break the uniqueness of the topology estimator. 
In addition, we demonstrated how to induce deliberate inaccuracies in the inferred topology even when solvability persists. 
A Laplacian structure-based feedback design was developed to meet the goal in a distributed manner. 
Numerical simulations further validated the effectiveness of our methods. 
Future work may consider the distributed design for the unsolvability and extend the analysis to more general network dynamics.

\bibliographystyle{IEEEtran}

\end{document}